\documentclass[11pt,dvipdfmx,a4paper]{article}
\usepackage[margin=25mm]{geometry}
\usepackage{amssymb,amsthm,enumerate,color}
\usepackage{amsmath,fancybox,amsfonts,ascmac,bm,ulem}
\usepackage{graphicx}
\usepackage{tikz}
\numberwithin{equation}{section}

\usepackage[colorlinks=true,linkcolor=blue,citecolor=blue]{hyperref}
\theoremstyle{definition}
\newtheorem{thm}{Theorem}[section]
\newtheorem{defn}[thm]{Definition}
\newtheorem{ex}[thm]{Example}

\newtheorem{lem}[thm]{Lemma}
\newtheorem{prop}[thm]{Proposition}

\newtheorem{rem}[thm]{Remark}

\newcommand{\R}{\mathbb{R}}   
\newcommand{\C}{\mathbb{C}}   
\newcommand{\N}{\mathbb{N}}   
\newcommand{\Z}{\mathbb{Z}}

\newcommand{\calC}{{\mathcal{C}}} 
\newcommand{\calD}{{\mathcal{D}}}

\newcommand{\calF}{{\mathcal{F}}}

\newcommand{\calM}{{\mathcal{M}}} 
\newcommand{\calN}{{\mathcal{N}}}

\newcommand{\dd}{{\rm d}}

\newcommand{\bn}{{\bm{n}}}
\newcommand{\bx}{{\bm{x}}}

\newcommand{\btau}{\boldsymbol \tau}

\title{Spectral analysis of Dirac operators for fermion scattering on topological solitons in the nonlinear $O(3)$ $\sigma$-model}

\author{Daiju Funakawa, Satoshi Okumura and Yuki Ueda}
\date{}

\begin{document}

\maketitle

\begin{abstract}
We investigate the existence of discrete positive or negative energy ground states of the Dirac operator $H$ which describe the fermion scattering on topological solitons in the nonlinear $O(3)$ $\sigma$-model. Additionally, we provide a sufficient condition to ensure that the positive and negative energies of the Dirac operator $H$ are non-zero.
\end{abstract}

%\tableofcontents

%%%%%%%
%SECTION 1
%%%%%%%

\section{Introduction}

Topological solitons of nonlinear field models play an important role in field theory, condensed matter physics and hydrodynamics etc. (see e.g. \cite{NT13, PS19}). In this paper, we study the Dirac operators for fermion scattering on topological solitons in the nonlinear $O(3)$ $\sigma$-model introduced by \cite{L21}. 

Let $\sigma_j \in M_2(\C)$ be the {\it Pauli matrices} for $j=1,2,3$, that is,
\[
\sigma_1=\left(\begin{array}{ccc}
      0 & 1 \\
      1 & 0 
    \end{array}\right), \qquad 
    \sigma_2=\left(\begin{array}{ccc}
      0 & -i \\
      i & 0 
    \end{array}\right), \qquad
    \sigma_3=\left(\begin{array}{ccc}
      1 & 0 \\
      0 & -1 
    \end{array}\right).
\]
The {\it spin matrices} are defined by $\gamma_j= (i\sigma_j ) \otimes 1_2$ for $j=1,2$ and $\gamma_3= (-\sigma_3 ) \otimes 1_2$, where $1_d$ is the $d\times d$ identity matrix. It is easy to verify that
\begin{itemize}
\item $\sigma_j^2=1_2$ for $j=1,2,3$ and also $\sigma_1\sigma_2=i\sigma_3$, $\sigma_2\sigma_3=i\sigma_1$ and  $\sigma_3\sigma_1=i\sigma_2$,
\item $\gamma_j^2 =1_4$ for $j=1,2,3$ and also $\gamma_1\gamma_2=i\gamma_3$, $\gamma_2\gamma_3=-i \gamma_1$ and $\gamma_3\gamma_1=-i\gamma_2$,
\item $\{\sigma_j, \sigma_k\}=0_2$ and $\{\gamma_j,\gamma_k\}=0_4$ for all $1\le j\neq k \le 3$,
\end{itemize}
where $\{a,b\}=ab+ba$ and $0_d$ is the $d\times d$ zero matrix.

Denote by $\calN$ the set of all vector fields $\bn=(n_1,n_2, n_3)$ such that $n_1,n_2,n_3$ are functions from $\R^2$ to $\R$ satisfying
\begin{enumerate}
\item[(n-1)] $n_j$ is continuously differentiable a.e. in $\R^2$ for $j=1,2,3$,
\item[(n-2)] $\|\bn(\bx)\|_{\R^3}=1$ a.e. $\bx \in \R^2$,
\item[(n-3)] $D_k n_j$ is bounded a.e. in $\R^2$ for all $j=1,2,3$ and $k=1,2$,
\item[(n-4)] $\|\nabla n_j(\bx)\|_{\R^2} \to 0$ as $\|\bx\|_{\R^2}\rightarrow \infty$ for $j=1,2,3$.
\end{enumerate}
Here we define the above notations as follows.
\begin{itemize}
\item $\|\cdot\|_{\R^n}$ is the standard Euclidian distance on $\R^n$ for $n\in \N$.
\item $D_k$ is the generalized partial differential operator in the variable $x_k$, acting in $L^2(\R^2)$ with domain
$$
\mathcal{D}(D_k):=\mathcal{D}\left(\overline{\frac{\partial}{\partial x_k}}\right), \quad k=1,2,
$$
where $\overline{T}$ is the closure of a closable operator $T$.
\item  $\nabla=(D_1,D_2)$ is the (generalized) gradient.
\end{itemize}
Due to condition (n-2), which asserts that $\text{esssup}_{\bx\in \R^2}|n_j(\bx)|\le 1$, the function $n_j$ can be considered a bounded multiplication operator on $L^2(\R^2)$. Additionally, in accordance with condition (n-3), the operator $D_k n_j$ can also be viewed as a bounded multiplication operator on $L^2(\R^2)$ for $j=1,2,3$ and $k=1,2$.

\begin{defn} 
Let us consider $m>0$, $\bn=(n_1,n_2,n_3)\in \calN$ and $\btau= (\tau_1,\tau_2,\tau_3):=(\sigma_1,\sigma_2,\sigma_3)$. A Dirac operator for {\it fermion scattering on topological solitons in the nonlinear $O(3)$ $\sigma$-model} (for short, {\it fermionic Dirac operator}) is the linear operator
\begin{align*}
 H_{\rm ferm} \equiv  H_{\rm ferm}(\bn):=-i \sum_{k=1}^2 \gamma_3 \gamma_k D_k \otimes 1_2 + m \sum_{j=1}^3 \gamma_3 n_j\otimes \tau_j,  
\end{align*}
acting in $L^2(\R^2;\C^4)\otimes \C^2$ with domain $\mathcal{D}(H_{\rm ferm})=\bigcap_{k=1}^2 (\mathcal{D}(D_k) \otimes \C^4) \otimes \C^2$.
\end{defn}
By the natural identification $L^2(\R^2;\C^4)\otimes \C^2 \cong \oplus_{i=1}^2 \left(L^2(\R^2;\C^2)\otimes \C^2\right)$, one can see that $H_{\rm ferm}$ has the matrix representation $H_{\rm ferm} = H\oplus H$, where 
\begin{align}\label{defn:OurDirac}
H\equiv H(\bn):= (-i\sigma_2  D_1 +i\sigma_1D_2) \otimes 1_2 - m \sum_{j=1}^3 \sigma_3 n_j \otimes \tau_j,
\end{align} 
acting in $L^2(\R^2;\C^2)\otimes \C^2$ with domain $\mathcal{D}(H)=\bigcap_{k=1}^2 (\mathcal{D}(D_k)\otimes \C^2) \otimes \C^2$. 
Since $H_{\rm ferm}$ is diagonal, we have $\sigma(H_{\rm ferm})=\sigma(H)$. Thus, it is enough to consider $H$ to study spectral properties of $H_{\rm ferm}$.

Denote by $\mathcal{C}=\{(r,\varphi): r>0,\ 0\le \varphi <2\pi\}$ the polar coordinate system and $\Phi:\mathcal{C}\rightarrow \R^2$ the change of coordinates. The unitary operator $U:L^2(\R^2)\rightarrow L^2(\mathcal{C}, r\dd r\dd\varphi )$ defined by $Uf = f\circ \Phi$ plays the role of the change of coordinates for operators. Then we get 
\begin{align}\label{eq:unitaryeqiv}
(U\otimes 1_2) H(\bn) (U\otimes 1_2)^\ast = (-i\sigma_2 \widetilde{D}_1 + i\sigma_1 \widetilde{D}_2) \otimes 1_2 - m \sum_{j=1}^3 \sigma_3 \widetilde{n}_j \otimes \tau_j,
\end{align}
where $\widetilde{n}_j=Un_jU^{-1}$ for $j=1,2,3$,
\begin{align*}
\widetilde{D}_1=\cos \varphi \cdot \partial_r- \frac{\sin\varphi}{r}\cdot \partial_\varphi,\qquad  \text{and} \qquad \widetilde{D}_2=\sin \varphi \cdot \partial_r + \frac{\cos\varphi}{r}\cdot \partial_\varphi,
\end{align*}
see Sections 3.1 and 3.2 for more details. Denote by $\widetilde{H}(\bn)$ the operator in the right hand side of $\eqref{eq:unitaryeqiv}$. The operators $H(\bn)$ and $\widetilde{H}(\bn)$ are unitarily equivalent and consequently have the same spectrum, that is, $\sigma(H(\bn))=\sigma(\widetilde{H}(\bn))$, $\sigma_{\rm d}(H(\bn))=\sigma_{\rm d}(\widetilde{H}(\bn))$ and $\sigma_{\rm ess}(H(\bn))=\sigma_{\rm ess}(\widetilde{H}(\bn))$, where $\sigma_{\rm d}(T)$ and $\sigma_{\rm ess}(T):=\sigma(T)\setminus \sigma_{\rm d}(T)$ are {\it discrete} and {\it essential spectrum} of a self-adjoint operator $T$, respectively.
Especially, we are interested in the Dirac operator $\widetilde{H}(\bn)$ when $\bn$ is a {\it hedgehog ansatz} $\bn(F)$, that is, for $\omega=(r,\varphi)\in \mathcal{C}$,
\begin{align*}
\widetilde{n}_1(\omega)=\sin F(r)\cos (N\varphi), \quad  \widetilde{n}_2(\omega)=\sin F(r)\sin (N\varphi), \quad \widetilde{n}_3(\omega)=\cos F(r),
\end{align*}
where $N\in \Z$ and $F: (0,\infty) \rightarrow [0,\pi]$ is a function satisfying the conditions (F-1)-(F-3) in Section 3.2. The function $F$ is called a {\it profile function}. One can see $\bn(F) \in \mathcal{N}$ in Lemma \ref{lem:HA}.

\begin{defn}
\begin{enumerate}
\item[\rm (i)] For a self-adjoint operator $T$ bounded from below, we define
$$
E_0(T):=\inf \sigma(T)
$$
 as the {\it lowest energy} of $T$. A non-zero vector in $\ker(T-E_0(T))$ is called a {\it discrete ground state of $T$} if $E_0(T) \in \sigma_{\rm d}(T)$. 
 
\item[\rm (ii)] For a self-adjoint operator $T$, we define
$$
E_0^+(T) := \inf (\sigma(T)\cap [0,\infty)) \quad \text{and} \quad E_0^-(T) := \sup (\sigma(T)\cap (-\infty,0])
$$
as the {\it positive} and {\it negative energy} of $T$, respectively. A non-zero vector in $\ker(T- E_0^+(T))$ (resp. $\ker(T- E_0^-(T))$)  is called a {\it discrete positive (resp. negative) energy ground state of $T$} if $E_0^+(T) \in \sigma_{\rm d}(T)$  (resp. $E_0^-(T) \in \sigma_{\rm d}(T)$).
\end{enumerate}
\end{defn}

\begin{rem}\label{rem1}
We understand $E_0^+(T) = \infty$ (resp. $E_0^-(T)=-\infty)$ when $\sigma(T)\cap[0,\infty)=\emptyset$ (resp. $\sigma(T)\cap (-\infty,0] =\emptyset$). If $\sigma(T)\neq \emptyset$, then either $E_0^+(T) \in [0,\infty)$ or $E_0^-(T) \in (-\infty,0]$ holds.
\end{rem}

In \cite{A06, AHS05}, they discussed the existence of discrete ground states in the context of the chiral quark soliton model. In this paper, we establish a sufficient condition for the existence of discrete positive or negative energy ground states of the fermionic Dirac operator $H(\bn(F))$. To derive the sufficient condition, our attention is directed towards the reduced component $H(\ell,s,t)$ of the fermionic Dirac operator $H(\bn(F))$. Specifically, we focus on its action within an eigenspace of the grand spin operator $K_3$, characterized by an eigenvalue of the form $\ell + \frac{s}{2} + \frac{Nt}{2}$, where $\ell$ is an integer, and $s$ and $t$ take values of $\pm 1$. Finally, our first main theorem is established. 

\begin{thm}
\label{main1}
Let us consider a profile function $F$. If there exist $\ell \in \Z$ and $t\in \{1,-1\}$ such that
\[
\inf_{\substack{f\in C_0^\infty((0,\infty))\\ \|f\|_2=1}} \left\langle f, \left(-\frac{\dd^2}{\dd r^2}-\frac{1}{r}\frac{\dd}{\dd r} +\frac{\ell^2+(\ell+1)^2}{2r^2}  + mt \frac{\dd F}{\dd r} \sin F\right)f \right\rangle_2<0,
\]
then $H(\bn(F))$ has a discrete positive or negative energy ground state, where $C_0^\infty((0,\infty))$ is the set of all compactly supported $C^\infty$-functions on $(0,\infty)$ and
$$
\langle u, v \rangle_2 := \int_{(0,\infty)} \overline{u(r)}v(r) r \dd r  \quad \text{and} \quad \|u\|_2:= \sqrt{\langle u, u\rangle_2}
$$
for all $u,v\in L^2((0,\infty), r\dd r)$.
\end{thm}
Our strategy of proof is to use the min-max principle as well as \cite{A06, AHS05}. However, the choice of vectors in the eigenspace of the grand spin operator $K_3$ is quite different from the previous works, see Lemma \ref{lem:Psi_f} for details. 

Next, we investigate a sufficient condition to ensure that the positive and negative energies of $H(\bn(F))$ are non-zero.

\begin{thm}\label{main2}
Let $\bn=\bn(F)$ be a hedgehog ansatz with $\limsup_{r\to 0^+}\left|\frac{\sin F(r)}{r}\right|<\infty$. If the assumption of Theorem \ref{main1} and
\[
m>\sup_{r>0} \sqrt{ \left|F'(r)^2 -\left(\frac{N}{r} \sin F(r)\right)^2 \right|}
\]
hold, then the positive and negative energies of $H=H(\bn(F))$ are non-zero, that is, $E_0^\pm (H)\neq 0$.
\end{thm}

The present paper is organized as follows. In Section 2, we investigate the discrete and essential spectrum of the fermionic Dirac operator $H(\bn)$. In Section 3.1, we explain differentiable operators under the polar coordinates. In Sections 3.2 and 3.3, we introduce and study the fermionic Dirac operator with a hedgehog ansatz. In Sections 4 and 5, we rigorously prove the main theorems (i.e. Theorems \ref{main1} and \ref{main2}) and provide an example (see Example \ref{ex:SUSY}).

%%%
%SECTION 2
%%%

\section{Discrete and essential spectrum}

In this section, we write $H$ as the operator $H(\bn)$ defined in \eqref{defn:OurDirac}. We explore the spectral properties of the fermionic Dirac operator $H$. We define the operator
$$
H_0:=(-i\sigma_2 D_1 +i\sigma_1D_2) \otimes 1_2
$$ 
acting in $L^2(\R^2;\C^2)\otimes \C^2$ with domain $\mathcal{D}(H_0)=\bigcap_{k=1}^2 (\mathcal{D}(D_k) \otimes \C^2)\otimes \C^2$. The operator $H_0$ is called the {\it spin-Dirac operator}. It is known that $H_0$ is a self-adjoint operator with $\mathcal{D}(H_0)$. Since the operator $m \sum_{j=1}^3 \sigma_3 n_j \otimes \tau_j$ is bounded and self-adjoint on $L^2(\R^2;\C^2)\otimes \C^2$, it follows from the Kato-Rellich theorem (see e.g. \cite[Theorem X.12]{RS75}) that $H=H_0- m \sum_{j=1}^3 \sigma_3 n_j \otimes \tau_j$ is self-adjoint with $\mathcal{D}(H)$. Therefore $\sigma(H)$ is contained in $\mathbb{R}$. 

In order to study the spectrum of $H$, we compute $H^2$ as follows.

\begin{lem}\label{lem:H^2}
For any $\bn =(n_1,n_2,n_3)\in \calN$, we have
$$
H^2 = (-\Delta +m^2)1_2\otimes 1_2 +V \quad \text{on} \quad \mathcal{D}(H^2),
$$
where $-\Delta:=-D_1^2 -D_2^2$ denotes the $2$-dimensional Laplacian with domain $\mathcal{D}(-\Delta) = \bigcap_{k=1}^2 \mathcal{D}(D_k^2)$ and
$$
V\equiv V(\bn):=  - m  \sum_{j=1}^3 \left(\sum_{k=1}^2 \sigma_k (D_k n_j)\right) \otimes \tau_j
$$
with domain $L^2(\R^2;\C^2)\otimes \C^2$. Moreover, $V$ is bounded, and therefore $\mathcal{D}(H^2)= \mathcal{D}(-\Delta)$.
\end{lem}

\begin{proof}
On the domain $\mathcal{D}(H^2)$, a direct computation shows
\begin{align*}
H^2 &= (i\sigma_2 D_1 \otimes 1_2)^2+  (i\sigma_1 D_2 \otimes 1_2)^2+ m^2 \left( \sum_{j=1}^3 \sigma_3n_j \otimes \tau_j\right)^2 \\
&\quad +  \underbrace{\{\sigma_2D_1, \sigma_1D_2\}}_{=0} \otimes 1_2 + im \sum_{j=1}^3 \{\sigma_2 D_1, \sigma_3 n_j\}\otimes \tau_j - im \sum_{j=1}^3 \{\sigma_1D_2, \sigma_3 n_j\}\otimes \tau_j\\
&=(-\Delta + m^2)1_2\otimes 1_2 + im \sum_{j=1}^3 \underbrace{\sigma_2\sigma_3}_{=i\sigma_1} (D_1n_j) \otimes \tau_j - im \sum_{j=1}^3  \underbrace{\sigma_1\sigma_3}_{=-i\sigma_2} (D_2n_j) \otimes \tau_j \\
&=(-\Delta+m^2)1_2\otimes 1_2 - m  \sum_{j=1}^3 \left(\sum_{k=1}^2 \sigma_k (D_k n_j)\right) \otimes \tau_j. 
\end{align*}
We get the desired formula on $\mathcal{D}(H^2)$. By the assumption (n-3), the operator $V$ is bounded on $ L^2(\R^2;\C^2)\otimes \C^2$, and therefore $\mathcal{D}(H^2)= \mathcal{D}(-\Delta)$. 
\end{proof}

For simplicity, we write $-\Delta+m^2$ instead of $(-\Delta+m^2)1_2\otimes 1_2$. Further, we define the following operator
\begin{align}\label{eq:VL}
L \equiv L(\bn):=-\Delta + V
\end{align}
with domain $\mathcal{D}(L)= \mathcal{D}(-\Delta)$. Then $H^2 = L + m^2$ holds on $\mathcal{D}(-\Delta)$.
Since $V$ is bounded and self-adjoint on $L^2(\R^2;\C^2)\otimes \C^2$, the operator $L$ is self-adjoint with domain $\calD(-\Delta)$ by Kato-Rellich theorem. 

By the definition of the $\sigma_k$, one can see the following matrix decomposition of $V$.
\begin{lem}\label{rem:Vdiagonal}
The operator $V$ is
$$
V=\left(\begin{array}{ccc}
      0 & W^*\\
     W & 0
     \end{array}\right),
$$
where $W\equiv W(\bn):= - m\sum_{j=1}^3 ((D_1n_j)+ i (D_2n_j)) \otimes \tau_j$.
\end{lem}

At the end of this section, we determine the explicit forms of $\sigma_{\rm d}(H)$ and $\sigma_{\rm ess}(H)$.

\begin{prop}\label{spectrumH}
For $m>0$, we have 
$$
\sigma_{\rm d}(H) = \sigma(H) \cap (-m,m) \quad \text{and} \quad \sigma_{\rm ess}(H)=\sigma(H) \cap \left( (-\infty, -m] \cup [m, \infty) \right).
$$
\end{prop}

\begin{proof}
By the assumption (n-4), we have $\|V(\bx)\| \to 0$ as $\|\bx\|\to \infty$. According to \cite[Theorem XIII.15 (b) in page 119]{RS78}, we get $\sigma_{\rm ess}(L)=[0,\infty)$. Hence $\sigma_{\rm ess}(H^2)=[m^2,\infty)$. This implies that $\sigma_{\rm ess}(H) \subset (-\infty,-m] \cup [m,\infty)$. Therefore $\sigma(H) \cap (-m,m) \subset \sigma_{\rm d}(H)$. On the other hand, if $\lambda\in \sigma_{\rm d}(H)$, then the spectral mapping theorem implies that $\lambda^2 \in \sigma_{\rm d}(H^2)\subset [0,m^2)$. Hence $\lambda \in \sigma(H) \cap (-m,m)$. Thus, $\sigma_{\rm d}(H)= \sigma(H) \cap (-m,m)$. It automatically implies that $\sigma_{\rm ess}(H) = \sigma(H) \cap\left( (-\infty, -m] \cup [m, \infty) \right)$. 
\end{proof}

%%%
%SECTION 3
%%%
\section{Dirac operators with hedgehog ansatz}

\subsection{Differential operators under the change of coordinates}
Recall that $\mathcal{C}:=\{(r,\varphi): r>0, 0\le \varphi<2\pi\}$. Define the change of coordinates $\Phi: \mathcal{C}\rightarrow \R^2$, that is, $\Phi (\omega):=(r\cos\varphi, r\sin\varphi)$ for $\omega=(r,\varphi)\in \mathcal{C}$. The map $\Phi$ is injective and $\Phi(\mathcal{C})=\R^2\setminus\{{\bf 0}\}$. For a complex-valued function $f$ on a domain which contains $\Phi(\mathcal{C})$ and $\omega \in \mathcal{C}$, we define 
$$
\widetilde{f}(\omega):=f(\Phi(\omega)).
$$
Then for any $\psi\in L^2(\R^2)$, we have
\[
\int_{\R^2}|\psi(\bx)|^2 \dd\bx=\int_{\Phi(\mathcal{C})} |\psi(\bx)|^2\dd\bx=\int_{\mathcal{C}}|\widetilde{\psi}(\omega)|^2 r\dd r \dd\varphi,
\]
and therefore 
$$
U: L^2(\R^2)\rightarrow L^2(\mathcal{C},r\dd r \dd\varphi ), \quad U\psi=\widetilde{\psi}
$$ 
is an isometry. Moreover, the operator $U$ is unitary. It is easy to see that $Ux_1U^{-1}= r\cos\varphi$ and $Ux_2U^{-1}=r\sin \varphi$, where $x_j$ is regarded as a multiplication operator acting in $L^2(\R^2)$ for $j=1,2$. By definitions of the change of coordinates and the chain rule, we see that
\begin{align*}
UD_1U^{-1}=\cos \varphi \cdot \partial_r -\frac{\sin\varphi}{r} \cdot \partial_\varphi \qquad \text{and} \qquad
UD_2U^{-1}=\sin \varphi \cdot \partial_r+\frac{\cos\varphi}{r} \cdot \partial_\varphi ,
\end{align*}
on $U \left(\bigcap_{k=1}^2\calD(D_k) \right)$, where $\partial_\#$ is a partial differential operator in the variable $\#=r,\varphi$. 

By the above relation, we have 
$$
-\widetilde{\Delta}:=U (-\Delta) U^{-1}= -\partial_r^2 - \frac{1}{r} \partial_r - \frac{1}{r^2} \partial_\varphi^2
$$ 
on $U\mathcal{D}(-\Delta)$. The unitary invariance implies $\sigma(-\widetilde{\Delta})=\sigma_{\rm ess}(-\widetilde{\Delta})=[0,\infty)$.

\subsection{Hedgehog ansatz}

Let $\mathcal{F}$ be the set of all functions $F:(0,\infty)\to [0,\pi]$ such that
\begin{enumerate}
\item[(F-1)] $F$ is continuously differentiable on $(0,\infty)$,
\item[(F-2)] $ \sup_{r>0 } |F'(r)|  <\infty$ and $F'\not\equiv 0$ on $(0,\infty)$,
\item[(F-3)] $\lim_{r\rightarrow \infty}F'(r) = 0$. 
\end{enumerate}
Recall the functions 
\begin{align*}
\widetilde{n}_1(\omega)=\sin F(r)\cos (N\varphi), \quad  \widetilde{n}_2(\omega)=\sin F(r)\sin (N\varphi), \quad \widetilde{n}_3(\omega)=\cos F(r)
\end{align*}
for $N\in \Z$, $F\in \mathcal{F}$ and $\omega =(r,\varphi)\in \mathcal{C}$. Define a vector field $\bn(F)=(n_1,n_2,n_3)$, where $n_j := U^{-1}\widetilde{n}_j U$ for $j=1,2,3$. The vector field $\bn(F)$ ($\cong (\widetilde{n}_1,\widetilde{n}_2,\widetilde{n}_3)$) is called a {\it hedgehog ansatz}. The number $N$ is known as the topological invariant, see \cite[Equation (4) and line 1 in page 1255]{ASI19}.

\begin{lem}\label{lem:HA}
For any $F\in \calF$, we have $\bn(F) \in \calN$.
\end{lem}

\begin{proof}
By definition of the change of coordinates, we get 
\begin{align*}
n_1(\bx) &= \sin F \left( \sqrt{x_1^2+x_2^2} \right) \cdot \cos (N \text{Tan}^{-1} (x_2/x_1)),\\ 
n_2(\bx) &= \sin F \left( \sqrt{x_1^2+x_2^2} \right)\cdot \sin (N \text{Tan}^{-1} (x_2/x_1)),\\
n_3(\bx) &= \cos F \left( \sqrt{x_1^2+x_2^2} \right)
\end{align*}
for all $\bx \in \R^2 \setminus\{(0,x_2): x_2\in\R^2\}$, where $\text{Tan}^{-1}: \R\to (-\frac{\pi}{2},\frac{\pi}{2})$ is the principal value of arctangent functions. Note that the line $\{(0,x_2): x_2\in\R\}$ is a zero set with respect to the Lebesgue measure on $\R^2$. By the assumptions (F-1) and (F-2), $n_j$ is continuously differentiable and $D_kn_j$ is bounded a.e. in $\R^2$ for $j=1,2,3$ and $k=1,2$. Thus $\bn(F)$ satisfies the conditions (n-1) and (n-3). From the definition of $\bn(F)$, it is easy to see the condition (n-2), that is, $\|\bn(F)(\bx)\|_{\R^3}=1$ a.e. $\bx \in \R^2$. By the chain rule, for $\omega=(r,\varphi)\in \mathcal{C}$, we have
\begin{align*}
\frac{\partial \widetilde{n}_1}{\partial r} (\omega)
&= \frac{\partial n_1}{\partial r} (r\cos\varphi, r\sin\varphi)\\
&=\frac{\partial n_1}{\partial x_1} \frac{\partial x_1}{\partial r} + \frac{\partial n_1}{\partial x_2} \frac{\partial x_2}{\partial r}=\frac{\partial n_1}{\partial x_1} \cos \varphi + \frac{\partial n_1}{\partial x_2} \sin \varphi.
\end{align*}
Since $\widetilde{n}_1(r,\varphi)=\sin F(r)\cos N\varphi$, it follows from the assumption (F-3) that
$$
\frac{\partial \widetilde{n}_1}{\partial r} (\omega) = F'(r) \cos F(r)\cos N\varphi \rightarrow 0, \qquad r\rightarrow\infty.
$$
Therefore we obtain $D_1 n_1(\bx)\rightarrow 0$ and $D_2 n_1(\bx)\rightarrow 0$ as $\|\bx\|_{\R^2}\rightarrow\infty$. Finally, we get
$\|\nabla n_1(\bx)\|_{\R^2} \rightarrow 0$ as $\|\bx\|_{\R^2}\rightarrow \infty$. Similarly, we obtain $\|\nabla n_2(\bx)\|_{\R^2} \rightarrow 0$ as $\|\bx\|_{\R^2}\rightarrow \infty$. Hence $\bn(F)$ satisfies the condition (n-4). Finally we get $\bn(F)\in \calN$. 
\end{proof}

For an operator $T$ acting in $L^2(\R^2;\C^2)\otimes \C^2$, we define the operator $\widetilde{T}$ acting in $L^2(\mathcal{C};\C^2)\otimes \C^2$ with domain $(U\otimes 1_2)\mathcal{D}(T)$ by
$$
\widetilde{T}:=(U\otimes 1_2) T (U\otimes 1_2)^{-1}.
$$  
We consider 
$$
H\equiv H(\bn(F)), \quad V\equiv V(\bn(F)) \quad \text{and} \quad L\equiv L(\bn(F))
$$ 
for $F\in \calF$. By Proposition \ref{spectrumH}, we have $\sigma_{\rm d} (\widetilde{H}) = \sigma(H)\cap (-m,m)$.
By Lemma \ref{lem:H^2}, we have 
$$
\widetilde{H}^2=\widetilde{L}+m^2=-\widetilde{\Delta} + \widetilde{V} + m^2,
$$ 
where 
\begin{align*}
\widetilde{V} =\left(\begin{array}{ccc}
      0 & m\sum_{j=1}^3 e^{-i\varphi} \left(\frac{\partial \widetilde{n}_j}{\partial r}  - \frac{i}{r} \frac{\partial \widetilde{n}_j}{\partial \varphi} \right) \otimes \tau_j\\
     -m\sum_{j=1}^3e^{i\varphi} \left(\frac{\partial \widetilde{n}_j}{\partial r}  + \frac{i}{r} \frac{\partial \widetilde{n}_j}{\partial \varphi} \right) \otimes \tau_j & 0
    \end{array}\right).
\end{align*}

\subsection{Reduced part with respect to the grand spin}

Let $L_3=\overline{-ix_1 D_2 + i x_2 D_1}$ be the third component of the angular momentum, which is the closure of the operator $-ix_1 D_2 + i x_2 D_1$. In the following, we regard $L_3$ as the operator $L_3\otimes 1_2$ acting in $L^2(\R^2;\C^2)$. Define the {\it grand spin} 
$$
K_3=L_3\otimes 1_2 + \frac{1}{2}\sigma_3 \otimes 1_2 + \frac{N}{2} \otimes \tau_3
$$
acting in $L^2(\R^2;\C^2)\otimes \C^2$. 
By the Kato-Rellich theorem, one can see that $K_3$ is a self-adjoint operator with domain $\mathcal{D}(K_3)=\mathcal{D}(L_3)\otimes \C^2$. By \cite[(5.16)]{AHS05},
$$
\sigma(K_3)=\sigma_{\rm p}(K_3)=\left\{ \ell + \frac{s}{2}+ \frac{Nt}{2} :\ \ell\in \Z,\ s, t \in \{1,-1\}\right\},
$$
where $\sigma_{\rm p}(T)$ is the point spectrum of $T$. By \cite[(5.17)]{AHS05}, the eigenspace of $K_3$ with an eigenvalue $\ell +(s/2)+(Nt/2)$ is given by 
$$
\calM_{\ell,s,t} := \ker(L_3-\ell) \otimes \ker(\sigma_3-s) \otimes \ker(\tau_3-t). 
$$
Therefore $L^2(\R^2;\C^2)\otimes \C^2$ has the following orthogonal decomposition 
$$
L^2(\R^2;\C^2)\otimes \C^2 \cong \bigoplus_{\ell\in \Z, s,t\in \{\pm 1\}} \calM_{\ell,s,t},
$$
under the natural identification $L^2(\R^2;\C^2)\otimes \C^2 \cong L^2(\R^2) \otimes \C^2\otimes \C^2$. 

According to \cite[Equation (14)]{L21}, we have $[H,K_3]=0$ i.e. $HK_3= K_3 H$. Moreover, a stronger relationship between $H$ and $K_3$ can also be expressed.

\begin{lem}\label{lem:M_{l,s,t}}
Consider $H\equiv H(\bn(F))$ for $F\in \mathcal{F}$.  
\begin{enumerate}
\item[\rm (i)] The operators $H$ and $K_3$ strongly commute, that is, their spectral projections commute.
\item[\rm (ii)] For any $\Psi\in \calD(H)$, we obtain $P_{\calM_{\ell,s,t}}\Psi \in \calD(H)$ and $HP_{\calM_{\ell,s,t}}\Psi=P_{\calM_{\ell,s,t}}H\Psi$, where $P_{\calM_{\ell,s,t}}$ is an orthogonal projection onto $\calM_{\ell,s,t}$.
\end{enumerate}
\end{lem}

\begin{proof}
The proof is almost the same as one in \cite[Lemmas 5.2 and 5.3]{AHS05}, but we include a proof for the reader's convenience. By \cite[Theorem VIII.13]{RS72}, if $H$ and $e^{it K_3}$ commute on $\calD(H)$ for all $t\in \R$, then $H$ and $K_3$ strongly commute. It is enough to show that $e^{-it K_3} H e^{it K_3}=H$ for all $t\in \R$ on the core $C_0^\infty(\R^2;\C^2)\otimes \C^2$ of $H$ since the equality can be extended to $\calD(H)$. First, we can obtain $e^{-it K_3}H_0 e^{it K_3} =H_0$ on $C_0^\infty(\R^2;\C^2)\otimes \C^2$ due to \cite[Lemma 5.2]{AHS05}. Next, we show that
\begin{align}\label{Conclusion}
E_t^* \left(\sum_{j=1}^3\beta \widetilde{n}_j \otimes \tau_j \right) E_t = \sum_{j=1}^3 \beta \widetilde{n}_j  \otimes \tau_j, \qquad \text{ on }  C_0^\infty(\mathcal{C};\C^2)\otimes \C^2,
\end{align}
where $E_t:=(U\otimes 1_2)e^{itK_3} (U\otimes 1_2)^{-1}$. Since $e^{it K_3}=e^{it L_3} e^{it \sigma_3/2} \otimes e^{itN \tau_3/2}$ for all $t\in \R$, we have
\begin{align}\label{eq:exp}
E_t=(Ue^{it L_3}U^{-1}) e^{it \sigma_3/2} \otimes e^{itN \tau_3/2}, \qquad t\in \R.
\end{align}
Moreover we get
\begin{align}\label{eq:shift}
\left(U e^{-it L_3} \cdot \widetilde{n}_j \cdot e^{it L_3} U^{-1}\psi \right)(\omega)=\widetilde{n}_j(r, \varphi-t)\psi(\omega),
\end{align}
for all $\omega=(r,\varphi)\in \mathcal{C}$ and $j=1,2,3$ and for all $\psi\in C_0^\infty(\mathcal{C};\C^2)$.
By \eqref{eq:exp} and \eqref{eq:shift}, we obtain
\begin{align}\label{eq:E(beta)E}
E_t^* \left(\sum_{j=1}^3\beta \widetilde{n}_j \otimes \tau_j \right)  E_t= \sum_{j=1}^3 \beta \widetilde{n}_j(r, \varphi-t) \otimes e^{-itN \tau_3/2} \tau_j e^{itN \tau_3/2}
\end{align}
on $C_0^\infty (\mathcal{C};\C^2) \otimes \C^2$.
Note that
\begin{align}\label{eq:tau1}
\tau_j e^{itN\tau_3/2} = e^{itN\tau_3/2} \tau_j e^{itN\tau_3} \qquad \text{ for each } j=1,2
\end{align}
and
\begin{align}
\tau_3 e^{itN\tau_3/2}= e^{itN\tau_3/2}\tau_3.
\end{align} 
Moreover, it follows from the proof of \cite[Lemma 5]{A06} that
\begin{align}\label{eq:tau123}
(\tau_1\cos Nt - \tau_2 \sin Nt  )e^{it N\tau_3} =\tau_1 \qquad \text{and} \qquad  (\tau_1\sin Nt +\tau_2\cos Nt)e^{itN\tau_3} = \tau_2.
\end{align}
By \eqref{eq:tau1}--\eqref{eq:tau123}, we obtain
\begin{align*}
\eqref{eq:E(beta)E}
&= \beta \sin F(r) \cos N\varphi  \otimes (\tau_1\cos Nt - \tau_2 \sin Nt ) e^{itN\tau_3} \\
&\hspace{4mm}+\beta \sin F(r)\sin N\varphi \otimes (\tau_1\sin Nt +\tau_2\cos Nt)e^{itN\tau_3}+\beta \cos F(r) \otimes \tau_3\\
&=\sum_{j=1}^3 \beta \widetilde{n}_j\otimes \tau_j
\end{align*}
on $C_0^\infty(\mathcal{C};\C^2)\otimes \C^2$. Hence we conclude that \eqref{Conclusion} holds on $C^{\infty}_0(\mathcal{C};\C^2)\otimes \C^2$, and therefore
$$
e^{-it K_3} \left(\sum_{j=1}^3\beta n_j\otimes \tau_j\right) e^{it K_3}=\sum_{j=1}^3 \beta n_j\otimes \tau_j,
$$
on $C^{\infty}_0 (\R^2;\C^2)\otimes \C^2$. Finally, $H$ and $K_3$ strongly commute. The second statement follows readily from the first one. 
\end{proof}

By Lemma \ref{lem:M_{l,s,t}}, one can define the reduced part $H(\ell,s,t)$ of $H$ to $\calM_{\ell,s,t}$, that is, $H(\ell,s,t):=H\upharpoonright_{\calD(H(\ell,s,t))}$, where $\calD(H(\ell,s,t))=\calD(H)\cap \calM_{\ell,s,t}$.
Note that $\calD(\widetilde{H}(\ell,s,t))=\calD(\widetilde{H})\cap  \widetilde{\calM}_{\ell,s,t}$, where $\widetilde{\calM}_{\ell,s,t}:=\ker(-i\partial_\varphi-\ell) \otimes \ker(\sigma_3-s)\otimes \ker(\tau_3-t)$. Since $UL_3U^{-1}=-i\partial_\varphi$, we get $-\frac{\partial^2}{\partial \varphi^2}\upharpoonright_{\calD(\widetilde{H}(\ell,s,t))}=\ell^2$, and therefore $\widetilde{\Delta} \upharpoonright_{\calD(\widetilde{H}(\ell,s,t))} = \frac{\dd^2}{\dd r^2} + \frac{1}{r} \frac{\dd}{\dd r} -\frac{\ell^2}{r^2}$.
In the following, we define the operator
$$
\Delta_\ell := \frac{\dd^2}{\dd r^2} + \frac{1}{r} \frac{\dd}{\dd r} -\frac{\ell^2}{r^2}, \qquad \ell \in \Z,
$$
with domain $\mathcal{D}(\Delta_\ell)=C_0^\infty((0,\infty))$. Moreover, for any $F\in \calF$, we get
\begin{align*}
\widetilde{V}\upharpoonright_{\calD(\widetilde{H}(\ell,s,t))} \ \equiv  \ \widetilde{V}(\bn(F))\upharpoonright_{\calD(\widetilde{H}(\ell,s,t))} \ =\left(\begin{array}{ccc}
      0 & \widetilde{W}(t)^*\\
      \widetilde{W}(t) & 0
    \end{array}\right),
\end{align*}
where 
$$
(\widetilde{W}(t)\Psi) (\omega):=- m e^{i\varphi} \left\{ \sum_{j=1}^2\left(\frac{\partial \widetilde{n}_j}{\partial r}  + \frac{i}{r} \frac{\partial \widetilde{n}_j}{\partial \varphi} \right) \otimes \tau_j + t \frac{\partial \widetilde{n}_3}{\partial r} \otimes 1_2 \right\} \Psi(\omega)
$$
for any $\Psi\in \calD(\widetilde{W}(t))$ and $\omega=(r,\varphi)\in \calC$.

%%%
%SECTION 4
%%%
\section{Discrete ground states}

In this section, we establish a sufficient condition for which $H\equiv H(\bn(F))$ has a discrete positive or negative energy ground state (see Theorem \ref{main1}). For later discussions, we set
$$
v_\ell (f) (\omega) := f(r) e^{i\ell \varphi}
$$ 
for all $\omega =(r,\varphi)\in \mathcal{C}$ and complex-valued functions $f$ on $(0,\infty)$ and $\ell\in \Z$. Let $u_t \in \C^2$ be a normalized eigenvector of $\tau_3$ for each $t\in \{1,-1\}$, that is, $\tau_3 u_t=t u_t$ for each $t\in \{1,-1\}$. Further, we consider $\ell \in \Z$ and  $f\in C_0^\infty((0,\infty))$ with $\|f\|_2=1$. We then define the following vectors:
\begin{align*}
\Psi_{f, +}^{(\ell, t)} =(v_\ell (f)\otimes u_t, 0)^\top \in \widetilde{\calM}_{\ell,1,t} \qquad \text{and} \qquad
\Psi_{f,-}^{(\ell,t)} =(0,v_\ell (f)\otimes u_t)^\top \in \widetilde{\calM}_{\ell,-1,t}.
\end{align*}
Moreover, for $\ell_1,\ell_2 \in \Z$ and $t\in \{1,-1\}$, we define
$$
\Psi_{f}^{(\ell_1,\ell_2,t)}:= \frac{1}{\sqrt{2}} \left(\Psi_{f,+}^{(\ell_1,t)}+\Psi_{f,-}^{(\ell_2,t)}\right).
$$ 
Note that $\|\Psi_{f}^{(\ell_1,\ell_2,t)}\|=1$.

\begin{lem}\label{lem:Psi_f}
Let us consider $\ell_1,\ell_2 \in \Z$, $t\in \{1,-1\}$ and $f\in C_0^\infty((0,\infty))$ with $\|f\|_2=1$. Moreover, $\widetilde{V}\equiv \widetilde{V}(\bn(F))$ for $F\in \calF$. Then we obtain the following formulas. 
\begin{enumerate}
\item[\rm (i)] $\displaystyle \langle \Psi_{f}^{(\ell_1,\ell_2,t)}, \widetilde{\Delta} \Psi_{f}^{(\ell_1,\ell_2,t)}\rangle = \frac{1}{2} \sum_{k=1}^2 \langle  f, \Delta_{\ell_k} f \rangle_2$.
\item[\rm (ii)] $\displaystyle \langle \Psi_{f}^{(\ell_1,\ell_2,t)}, \widetilde{V} \Psi_{f}^{(\ell_1,\ell_2,t)}\rangle = mt \langle f , (F' \sin F) f \rangle_2 \  \delta_{\ell_1+1, \ell_2}$.
\end{enumerate}
In particular, for each $\ell \in \Z$, we have
\begin{align*}
\langle \Psi_{f}^{(\ell,\ell+1,t)}, \widetilde{L} \Psi_{f}^{(\ell,\ell+1,t)} \rangle = \langle f, R(\ell,t) f\rangle_2,
\end{align*}
where we recall $\widetilde{L} = -\widetilde{\Delta} + \widetilde{V}$ on $\mathcal{D}(-\widetilde{\Delta})$ and
$$
R (\ell ,t) :=- \frac{1}{2}\Delta_\ell - \frac{1}{2}\Delta_{\ell+1} + mt \frac{\dd F}{\dd r} \sin F,
$$
is a symmetric operator acting in $L^2((0,\infty), r\dd r)$ with domain $\calD(R (\ell ,t) ) =C_0^\infty((0,\infty))$.
\end{lem}

\begin{proof}
(i) Since $\langle \Psi_{f,+}^{(\ell_1,t)}, \widetilde{\Delta} \Psi_{f,-}^{(\ell_2,t)} \rangle=\langle \Psi_{f,-}^{(\ell_2,t)}, \widetilde{\Delta} \Psi_{f,+}^{(\ell_1,t)} \rangle =0$, we obtain
\begin{align*}
\langle \Psi_{f}^{(\ell_1,\ell_2,t)}, \widetilde{\Delta} \Psi_{f}^{(\ell_1,\ell_2,t)}\rangle
&= \frac{1}{2} \left( \langle \Psi_{f,+}^{(\ell_1,t)}, \widetilde{\Delta} \Psi_{f,+}^{(\ell_1,t)} \rangle  + \langle \Psi_{f,-}^{(\ell_2,t)}, \widetilde{\Delta} \Psi_{f,-}^{(\ell_2,t)} \rangle  \right)\\
&= \frac{1}{2} \left( \langle \Psi_{f,+}^{(\ell_1,t)}, \Delta_{\ell_1} \Psi_{f,+}^{(\ell_1,t)} \rangle +  \langle \Psi_{f,-}^{(\ell_2,t)}, \Delta_{\ell_2} \Psi_{f,-}^{(\ell_2,t)} \rangle \right)\\
&= \frac{1}{2} \left( \langle f, \Delta_{\ell_1} f \rangle_2 +  \langle f, \Delta_{\ell_2} f \rangle_2 \right).
    \end{align*}

(ii) We have
\begin{align*}
\langle \Psi_{f,+}^{(\ell_1,t)}, \widetilde{V} \Psi_{f,+}^{(\ell_1,t)} \rangle 
&= \left\langle 
\left(\begin{array}{c}
    v_{\ell_1}(f)\otimes u_t\\
    0\\
    \end{array}\right),
     \left(\begin{array}{ccc}
      0 & \widetilde{W}(t)^*\\
      \widetilde{W}(t) & 0
    \end{array}\right) \left(\begin{array}{c}
    v_{\ell_1}(f)\otimes u_t\\
    0\\
    \end{array}\right) \right\rangle\\
&=\left\langle 
\left(\begin{array}{c}
    v_{\ell_1}(f)\otimes u_t\\
    0\\
    \end{array}\right),
     \left(\begin{array}{c}
    0\\
    \widetilde{W}(t)(v_{\ell_1}(f)\otimes u_t)\\
    \end{array}\right) \right\rangle\\
&=0.
\end{align*}
Similarly, we obtain $\langle \Psi_{f,-}^{(\ell_2,t)}, \widetilde{V} \Psi_{f,-}^{(\ell_2,t)} \rangle =0$. Hence,
\begin{align*}
\langle \Psi_{f}^{(\ell_1,\ell_2,t)},& \widetilde{V} \Psi_{f}^{(\ell_1,\ell_2,t)}\rangle \\
&= \frac{1}{2} \left(\langle \Psi_{f,+}^{(\ell_1,t)}, \widetilde{V} \Psi_{f,-}^{(\ell_2,t)} \rangle + \langle \Psi_{f,-}^{(\ell_2,t)}, \widetilde{V} \Psi_{f,+}^{(\ell_1,t)} \rangle  \right)\\
&=\frac{1}{2}\left( \langle v_{\ell_1}(f)\otimes u_t, \widetilde{W}(t)^* (v_{\ell_2}(f)\otimes u_t)\rangle  + \langle v_{\ell_2}(f)\otimes u_t, \widetilde{W}(t) (v_{\ell_1}(f)\otimes u_t)\rangle\right)\\
&=\text{Re}\langle \widetilde{W}(t)(v_{\ell_1}(f)\otimes u_t), v_{\ell_2}(f)\otimes u_t\rangle\\
&= - mt\text{Re}  \left(\int_{0}^{2\pi}e^{-i (1+\ell_1)\varphi}e^{i\ell_2\varphi} \cdot \frac{1}{2\pi} \dd\varphi \right) \left(\int_{(0,\infty)}  \frac{\dd \cos F(r)}{\dd r} \overline{f(r)} f(r)r\dd r \right)\\
&= mt \langle f, (F' \sin F) f \rangle_2 \ \delta_{\ell_1+1, \ell_2}.
\end{align*}
We get the above formulas. By (i) and (ii), for each $\ell \in \Z$, we have
\begin{align*}
\langle \Psi_{f}^{(\ell,\ell+1,t)}, \widetilde{L} \Psi_{f}^{(\ell,\ell+1,t)} \rangle &=  \langle  f,  -\frac{1}{2} \Delta_{\ell} f \rangle_2 + \langle  f,  -\frac{1}{2} \Delta_{\ell+1} f \rangle_2 + \langle f, (mt F' \sin F)f \rangle_2\\
&=\langle f, R(\ell,t) f\rangle_2,
\end{align*}
as desired. Since the operators $\Delta_\ell$ and $\Delta_{\ell+1}$ are symmetric with domain $C_0^\infty((0,\infty))$, so is $R(\ell,t)$.
\end{proof}

In order to prove Theorem \ref{main1}, we prepare a general result on the lowest energy.
\begin{lem}\label{lem:lowest_energy}
Let $T$ be a self-adjoint operator. Then we have $E_0(T^2) \ge \min \{E_0^+(T), |E_0^-(T)|\}^2$.
\end{lem}

\begin{proof}
Since $\sigma(T)\neq \emptyset$, we have $E_0^+(T) \in [0,\infty)$ or $E_0^-(T) \in (-\infty,0]$ by Remark \ref{rem1}. We define $a= \min \{E_0^+(T), |E_0^-(T)|\} \in [0,\infty)$.  One can see that $\sigma(T) \subset (-\infty, E_0^-(T)] \cup [E_0^+(T), \infty)$, where we understand $[\infty,\infty)=\emptyset$ (resp. $(-\infty, -\infty]=\emptyset$) if $E_0^+(T)=\infty$ (resp. $E_0^-(T)=-\infty$). By the spectral mapping theorem, we obtain $\sigma(T^2) \subset [a^2,\infty)$. This implies that $E_0(T^2) \ge a^2$. 
\end{proof}

Finally, we restate Theorem \ref{main1} and present its proof.

\begin{thm}\label{thm:mainFGS}
Let us consider $F\in \calF$. If there exist $\ell \in \Z$ and $t\in \{1,-1\}$ such that
\begin{align}\label{eq:E_0condi}
\mathcal{E}_0\left(R(\ell, t)\right):= \inf_{\substack{f\in C_0^\infty ((0,\infty)) \\ \|f\|_2=1}} \langle f ,  R(\ell, t) f \rangle_2 <0,
\end{align}
then $H$ has a discrete positive or negative energy ground state.
\end{thm}

\begin{proof}
If $\mathcal{E}_0\left(R(\ell, t)\right)<0$ for some $\ell \in \Z$ and $t\in \{1,-1\}$, then there exists $f\in C_0^\infty((0,\infty))$ with $\|f\|_2=1$ such that $\langle f, R(\ell, t)f \rangle_2 <0$. By Lemma \ref{lem:Psi_f}, we have
$$
\langle \Psi_{f}^{(\ell,\ell+1,t)}, \widetilde{L} \Psi_{f}^{(\ell,\ell+1,t)}\rangle =\langle f, R(\ell, t)f \rangle_2 <0. 
$$
Since $\widetilde{L}$ is self-adjoint and bounded from below, the above inequality implies that
$$
E_0(\widetilde{L}) \le \langle \Psi_{f}^{(\ell,\ell+1,t)}, \widetilde{L} \Psi_{f}^{(\ell,\ell+1,t)}\rangle<0.
$$
By the proof of Proposition \ref{spectrumH}, we notice that $\sigma_{\rm ess} (\widetilde{L}) = [0,\infty)$. It follows from the min-max principle (see \cite[Theorem XIII.2]{RS78})  that $E_0(\widetilde{L}) \in \sigma_{\rm d}(\widetilde{L})$. Therefore, we obtain 
\begin{align}\label{eq:E0(H^2)}
E_0(\widetilde{H}^2) =E_0(\widetilde{L})+m^2< m^2.
\end{align}
By Lemma \ref{lem:lowest_energy}, we get $E_0(\widetilde{H}^2) \ge \min\{E_0^+(\widetilde{H}), |E_0^-(\widetilde{H})|\}^2$. First, we consider the case of that $E_0^+(\widetilde{H})=\min\{E_0^+(\widetilde{H}), |E_0^-(\widetilde{H})|\}$. Then we assume here that $E_0^+ (\widetilde{H}) \notin \sigma_{\rm d}(\widetilde{H})$, i.e. $E_0^+(\widetilde{H}) \in \sigma_{\rm ess}(\widetilde{H})$. Then $E_0^+(\widetilde{H})^2 \in \sigma_{\rm ess}(\widetilde{H}^2)=[m^2,\infty)$ by the spectral mapping theorem. That is, $E_0(\widetilde{H}^2) \ge m^2$. However, this contradicts \eqref{eq:E0(H^2)}. Therefore $E_0^+(\widetilde{H})\in \sigma_{\rm d}(\widetilde{H})$. Similarly, we can conclude that $E_0^-(\widetilde{H}) \in \sigma_{\rm d}(\widetilde{H})$ if $|E_0^-(\widetilde{H})|=\min\{E_0^+(\widetilde{H}), |E_0^-(\widetilde{H})|\}$.
Thus $\widetilde{H}$ has a discrete positive or negative energy ground state, and therefore also $H$. 
\end{proof}

%%%%%%
%SECTION 5
%%%%%%

\section{Non-zero positive and negative energies}

In this section, we give a sufficient condition for which $E_0^\pm (H)\neq 0$ (see Theorem \ref{main2}). To obtain the condition, we prepare a few technical lemmas.

\begin{lem}\label{lem:MINENERGY}
Under the assumption of Theorem \ref{thm:mainFGS}, we obtain
$$
E_0(H^2) = \min \{E_0^+(H), |E_0^-(H)|\}^2.
$$
\end{lem}
\begin{proof}
We define $a= \min \{E_0^+(H), |E_0^-(H)|\}<\infty$. According to Lemma \ref{lem:lowest_energy}, we have already obtained $E_0(H^2)\ge a^2$. By Theorem \ref{thm:mainFGS}, we have $E_0^+(H) \in \sigma_{\rm d}(H)$ or $E_0^-(H) \in \sigma_{\rm d}(H)$, and hence $a^2 \in \sigma_{\rm d}(H^2)$. Therefore we get $a^2 \ge  E_0(H^2)$, as desired. 
\end{proof}

\begin{lem}\label{lem:spec_V}
Let us set
\begin{align*}
V_0=\left(\begin{array}{ccc}
      0 & W_0^*\\
     W_0 & 0
    \end{array}\right),
\end{align*}
where $W_0$ is a bounded operator on some Hilbert space. Then we have $\|V_0\|=\|W_0\|$.
\end{lem} 
\begin{proof}
One can see that $V_0^2=(W_0^\ast W_0) \oplus (W_0 W_0^\ast)$. Therefore $V_0$ is bounded and self-adjoint. By \cite[Proposition 4.4 in page 155]{A18}, we obtain $\|V_0^2\|= \max\{\|W_0^\ast W_0\|, \|W_0W_0^\ast\|\}$. Since it is known that $\|T^\ast T\|=\|TT^\ast\|=\|T\|^2$ for any bounded operators $T$ on a Hilbert space, we notice $\|V_0\|^2 =\|W_0\|^2$, as desired.
\end{proof}

We recall from Lemma \ref{rem:Vdiagonal} that for $\bn=(n_1,n_2,n_3)\in \mathcal{N}$,
$$
H^2=-\Delta+m^2+  \underbrace{\left(\begin{array}{ccc}
      0 & W^*\\
     W & 0
    \end{array}\right)}_{=V},
$$
where $W= - m \sum_{j=1}^3 ((D_1n_j) + i(D_2n_j)) \otimes \tau_j$. If $\bn=\bn(F)$ is a hedgehog ansatz with $F\in \calF$ and 
$\limsup_{r\to 0^+} \left|\frac{\sin F(r)}{r}\right|<\infty$, then 
$$
\widetilde{W}:=(U\otimes 1_2)W(U\otimes 1_2)^{-1}= - m \left(\begin{array}{ccc}
     \widetilde{R}_1 & \widetilde{R}_2\\
     \widetilde{R}_3 & -\widetilde{R}_1
    \end{array}\right),
$$
where 
\begin{align*}
\widetilde{R}_1&=-F'(r) \sin F(r) e^{i\varphi},\\
\widetilde{R}_2&= \left( F'(r)\cos F(r) + \frac{N}{r} \sin F(r)\right) e^{-i(N-1)\varphi},\\
\widetilde{R}_3&=  \left(F'(r)\cos F(r) - \frac{N}{r} \sin F(r) \right) e^{i(N+1)\varphi}.
\end{align*}
One can see that
\begin{align}\label{eq:norm_R}
\| \widetilde{W} \|= m  \sup_{r>0} \sqrt{\left|F'(r)^2 -\left(\frac{N}{r} \sin F(r)\right)^2 \right|}.
\end{align}

Finally, we restate Theorem \ref{main2} and present its proof.

\begin{thm}
Let us consider $F \in \calF$ with $\limsup_{r\to 0^+} \left|\frac{\sin F(r)}{r}\right|<\infty$.  If the assumption of Theorem \ref{thm:mainFGS} and
\begin{align}
\label{eq:condi_SUSY}
m>\sup_{r>0} \sqrt{\left|F'(r)^2 -\left(\frac{N}{r} \sin F(r)\right)^2 \right|}
\end{align}
hold, then $E_0^\pm (H) \neq 0$.
\end{thm}

\begin{proof}
By Lemma \ref{lem:spec_V} and \eqref{eq:norm_R}, we obtain
\begin{align*}
E_0(H^2) &= \inf_{\substack{\Phi\in\mathcal{D}(H^2)\\ \|\Phi\|=1}} \langle \Phi, H^2 \Phi \rangle\\
&\ge  m^2 +\inf_{\substack{\Phi\in\mathcal{D}(H^2)\\ \|\Phi\|=1}}\langle \Phi, V\Phi \rangle\\
&\ge m^2 - \|V\|\\
&=m^2 - \|\widetilde{W}\| \\
&=m^2 - m  \sup_{r>0} \sqrt{\left|F'(r)^2 -\left(\frac{N}{r} \sin F(r)\right)^2 \right|}.
\end{align*}
The assumption \eqref{eq:condi_SUSY} shows that $E_0(H^2)>0$.  By Lemma \ref{lem:MINENERGY}, the desired result holds. 
\end{proof}

We provide an example of $F\in \mathcal{F}$ such that  $H=H(\bn(F))$ has a discrete positive or negative energy ground state and $E_0^\pm (H)\neq 0$. More precisely, we find an $F\in \calF$ satisfying $\limsup_{r\to 0^+} \left|\frac{\sin F(r)}{r}\right|<\infty$ and the two conditions \eqref{eq:E_0condi} and \eqref{eq:condi_SUSY} by Theorems \ref{main1} and \ref{main2}. In order to find such an $F\in \mathcal{F}$, we define the following mollifier
\begin{align}\label{eq:testfunction}
f_0(r):= Cre^{-\frac{1}{1-r^2}}\chi_{\{0<r<1\}}(r) \quad \text{and} \quad C:= \left(\int_0^1 r^3 e^{-\frac{2}{1-r^2}}\dd r\right)^{-\frac{1}{2}},
\end{align}
where $\chi_A$ is the indicator function on $A \subset (0,\infty)$.
One can see that the function $f_0$ is in $C_0^\infty((0,\infty))$ and $\|f_0\|_2=1$.

\begin{ex}\label{ex:SUSY}
\label{ex:num}
We define
$$
F(r) = \frac{\pi}{r+1}, \qquad r>0.
$$ 
One can see that $\limsup_{r\to 0^+} \left|\frac{\sin F(r)}{r}\right|<\infty$. Consequently, if we put $N=1$, then
\[
\sup_{r>0}  \sqrt{ \left|F'(r)^2 -\left(\frac{N}{r} \sin F(r)\right)^2 \right|}  \approx 1.223.
\]
To achieve the condition \eqref{eq:condi_SUSY}, we give $m\ge 1.23$.

Next, we choose $m$ such that 
$$
\frac{1}{2}<\frac{9\pi}{200}m-\frac{5}{2},
$$ 
which implies $m>\frac{200}{3\pi}\approx 21.22$. Additionally, we pick up $\ell\in \Z$ such that $\ell^2+(\ell+1)^2 <\frac{9\pi}{100}m-5$.
The existence of such an $\ell$ is guaranteed due to the chosen value of $m$ and the fact that $\ell^2+(\ell+1)^2 \ge 1$ holds for all $\ell \in \Z$. For the function $f_0$ defined in \eqref{eq:testfunction}, 
\[
\left(- \frac{1}{2}\Delta_\ell - \frac{1}{2}\Delta_{\ell+1}\right) f_0 (r)= \left\{ \frac{-4(3r^2-2)}{(1-r^2)^4} + \frac{\ell^2+(\ell+1)^2-2}{2r^2}\right\} f_0(r),
\]
and therefore $\langle f_0, \left(- \frac{1}{2}\Delta_\ell - \frac{1}{2}\Delta_{\ell+1}\right) f_0\rangle_2  \le 15+3\ell^2 + 3(\ell+1)^2$. Moreover, if we put $t=1$, then
\begin{align*}
\langle f_0, mt F' \sin F f_0\rangle_2 
&= \left\langle f_0, -\frac{m\pi}{(r+1)^2} \sin \left(\frac{\pi}{r+1}\right) f_0 \right\rangle_2\\
&=-m\pi \underbrace{\left(\int_0^1 r^3 e^{-\frac{2}{1-r^2}}\dd r\right)^{-1} }_{\approx 270.23} \underbrace{\int_0^1 \frac{r^3}{(r+1)^2}\sin\left(\frac{\pi}{r+1}\right) e^{-\frac{2}{1-r^2}}\dd r}_{\approx 0.00135476}\\
&\le - \frac{27 \pi}{100}m.
\end{align*}
With the appropriate choice of values for $m$ and $\ell$, we ultimately obtain
\begin{align*}
\mathcal{E}_0\left(- \frac{1}{2}\Delta_\ell - \frac{1}{2}\Delta_{\ell+1}  + mtF'\sin F\right)  
&\le \left\langle f_0, \left(- \frac{1}{2}\Delta_\ell - \frac{1}{2}\Delta_{\ell+1} + mtF'\sin F\right)f_0 \right\rangle_2\\
&=15 + 3\ell^2 + 3(\ell+1)^2 - \frac{27 \pi}{100}m\\
&<0.
\end{align*}

The conditions \eqref{eq:E_0condi} and \eqref{eq:condi_SUSY} are satisfied, especially when $m\ge 21.23$.
\end{ex}

\subsection*{Acknowledgement} 

This work was started when S.O. was affiliated with Department of Mathematics, Tohoku University. D.F. and Y.U. thank Kazuyuki Wada for fruitful discussions. The authors are grateful to referees for very helpful comments to improve the readability of paper.

\bibliographystyle{amsplain}

\vspace{8mm}

\hspace{-6mm}{\bf Daiju Funakawa}\\
Department of Electronics and Information Engineering, Hokkai-Gakuen University, Sapporo, Hokkaido 062-8605, Japan
E-mail: funakawa@hgu.jp\\

\vspace{3mm}

\hspace{-6mm}{\bf Satoshi Okumura}\\
Yokohama, Kanagawa, Japan\\
E-mail: s.okumura.tohoku@gmail.com\\

\vspace{3mm}

\hspace{-6mm}{\bf Yuki Ueda}\\
Department of Mathematics, Hokkaido University of Education, 9 Hokumon-cho, Asahikawa, Hokkaido 070-8621, Japan\\
E-mail: ueda.yuki@a.hokkyodai.ac.jp

\end{document}